\documentclass{scrartcl}
\KOMAoptions{paper=letter}

\usepackage[fixcref]{group} 
\usepackage{multirow}
\usepackage{arydshln}

\theoremstyle{remark} 
\newtheorem{step}{Step}
\crefname{step}{step}{steps}
\Crefname{step}{Step}{Steps}
\newtheorem{claim}{Claim}
\crefname{claim}{Claim}{Claims}
\Crefname{claim}{Claim}{Claims}

\newcommand{\uni}{\operatorname{uni}}

\newcommand{\prof}[1][]{\ifthenelse{\equal{#1}{}}{\mathcal{R}}{\mathcal{R}_{#1}}}
\newcommand{\fprof}[1][]{\ifthenelse{\equal{#1}{}}{\mathcal{R}^*}{\mathcal{R}^*_{#1}}}

\newcommand{\fone}{\mathcal{F}}
\renewcommand{\pref}{\succcurlyeq}

\newcommand{\conv}{\mathrm{conv}}

\renewcommand{\int}{\mathrm{int}}
\newcommand{\relint}{\mathrm{relint}}

\renewcommand{\epsilon}{\varepsilon}

\title{Consistent Probabilistic Social Choice Revisited}

\author{Florian Brandl\\University of Bonn%
\and 
Felix Brandt\\
Technical University of Munich%
}

\begin{document}

\maketitle

\begin{abstract}
\citet{Bran13a} characterized a probabilistic social choice function known as maximal lotteries within a framework based on fractional preference profiles, which abstracts away from individual voters.
While this modeling assumption enables a more elegant and transparent proof, it complicates comparison with other results in the literature.
The purpose of this note is to transfer their results to the standard model of social choice, where each preference profile is defined for a finite number of voters.
Along the way, we prove a slightly stronger version of their main theorem that uses a weaker continuity condition and allows for real-valued (rather than only rational-valued) probabilities.
\end{abstract}

\section{Introduction}

\citet{Bran13a} have shown that a probabilistic social choice function (PSCF) known as maximal lotteries is characterized by population-consistency and composition-consistency. The characterization was proved in a framework that operates on \emph{fractional profiles}, i.e., preference profiles are modeled as functions that associate each possible preference relation with the fraction of voters who share these preferences. While this framework is not without precedent \citep[see, e.g.,][]{Youn74b,Youn75a,YoLe78a,Saar95a,Myer95b,DaMa08a}, it differs from the standard model of social choice theory based on the preferences of individual voters. When PSCFs are defined for fractional profiles, they are \emph{anonymous} (i.e., all voters are treated identically) and \emph{homogeneous} (i.e., replicating the electorate does not affect the outcome). 
The characterization by \citet{Bran13a} specifically rests on the underlying framework of fractional profiles when defining \emph{decisiveness} (the set of profiles where the PSCF is resolute is dense in the set of all preference profiles) and \emph{continuity} (the PSCF is upper hemi-continuous).

The purpose of this note is to transfer the theorem by \citet{Bran13a} to the standard model of social choice, where each preference profile is defined for a finite number of voters. This facilitates comparison with other results in the literature. 
In addition to the main theorem by \citeauthor{Bran13a}, we also transfer their alternative characterization using weak Condorcet-consistency and an impossibility result for non-probabilistic PSCFs to the standard setting. These translations are achieved by proving a slightly stronger version of the theorem by \citeauthor{Bran13a} that replaces upper hemi-continuity with a weaker continuity condition. Furthermore, we allow for lotteries with real-valued (rather than only rational-valued) probabilities in the main theorem.

\section{Preliminaries}\label{sec:prelims}

Let $U$ be an infinite universal set of alternatives.
The set of \emph{agendas} from which alternatives are to be chosen is the set of finite and non-empty subsets of~$U$, denoted by~$\fone(U)$.
The set of all linear (i.e., complete, transitive, and antisymmetric) \emph{preference relations} over some set $A\in\fone(U)$ will be denoted by $\mathcal{L}(A)$.
The set of all preference profiles for an agenda $A$ is denoted by $\prof[A]=\mathbb N_0^{\mathcal L(A)}\setminus\{0\}$, where $\mathbb N_0 = \mathbb N\cup\{0\}$ and $0$ denotes the empty profile without voters, and $\mathcal{R} = \bigcup_{A\in\fone(U)} \prof[A]$. 
Hence, preference profiles are \emph{anonymous} in the sense that they only represent how many voters hold each preference relation.
For a finite set $X$, we denote by $\Delta(X)$ the set of all probability distributions over $X$.
We interpret $R({\pref})$ as the number of voters with preference relation ${\pref}\in\mathcal{L}(A)$.
With slight abuse of notation, we also write ${\pref}$ for the profile that consists of a single voter with preference relation ${\pref}$.
Preference profiles are depicted by tables in which each column represents a preference relation $\pref$ with $R({\pref})>0$.
For $B\subseteq A$ and $R\in\prof[A]$, $R|_B$ is the restriction of $R$ to alternatives in $B$, i.e., for all ${\pref}\in\mathcal{L}(B)$,

\[R|_B({\pref}) = \sum_{{\pref}'\in \mathcal{L}(A)\colon{\pref}\subseteq{\pref}'} R({\pref}')\text.\]
For all $x,y\in A$, $R(x,y) = R|_{\{x,y\}}({\{(x,y)\}})$ is the number of voters who prefer $x$ to $y$ (the set $\{(x,y)\}$ represents the preference relation on two alternatives with $x\succ y$). 
Elements of $\Delta(A)$ are called \emph{lotteries} and will be written as convex combinations of alternatives. 
If $p$ is a lottery, $p_x$ is the probability that $p$ assigns to alternative $x$. 

A \emph{probabilistic social choice function (pSCF)} $f$ is a function that, for any agenda $A\in\fone(U)$, maps a preference profile $R\in \prof[A]$ to a non-empty and convex subset of $\Delta(A)$, and satisfies the following four properties.\footnote{We write PSCF (uppercase) for probabilistic social choice functions on fractional profiles, the framework of \citet{Bran13a}, and pSCF (lowercase) for those on the standard model with a finite number of voters, which is the setting of our results.}

\newlength{\firstcol}
\settowidth{\firstcol}{\emph{Homogeneity:}}

\medskip
\noindent
\begin{tabular}{@{}p{\firstcol}p{\linewidth-\firstcol-\tabcolsep-\tabcolsep}@{}} 
\emph{Homogeneity:} & For all $A\in\fone(U)$, $R\in\prof[A]$, and $k\in\mathbb N$,
$
    f(R) = f(kR)
$.\\
\emph{Faithfulness:} & For all $x,y\in U$, 
$
    R=\{(1,(x,y))\} \text{ implies } f(R) \neq \{y\}
$.\\
\emph{Continuity:} & For all $A\in\fone(U)$, $R,R'\in\prof[A]$, and $p\in\Delta(A)$, if
$f(kR+R') = \{p\}$ for all $k\in\mathbb N$, then $p \in f(R)$.\\
\emph{Resolvability:} & For all $A\in\fone(U)$ and $R \in \prof[A]$ with $|f(R)| > 1$, there is ${\pref} \in\mathcal L(A)$ such that
$
    \lvert f(R + {\pref})\rvert = 1
$.
\end{tabular}
\medskip

Let us discuss each of these assumptions individually.
Homogeneity is well-known in social choice theory \citep[see, e.g.,][]{Smit73a,Youn77a}.
It demands that replicating the entire electorate does not affect the outcome of the pSCF. In the context of fair division, this property is sometimes called replication invariance.\footnote{The ``inclusion'' part of homogeneity ($f(R) \subseteq f(kR)$) also follows from population-consistency defined below. Thus, it is sufficient to demand the ``exclusion'' part of homogeneity ($f(R) \supseteq f(kR)$).}

Faithfulness is an extremely weak condition concerning the special case of a single voter and two alternatives. The original version, due to \citet{Youn74a,Youn75a}, requires that if the voter prefers $x$ to $y$, $x$ should be selected. Our version is even weaker and only demands that $f$ does not exclusively return the lottery that selects $y$ with probability 1. Among the four conditions we demand from pSCFs, it is the only one that interprets the preference relation.

Continuity goes back to \citet{Youn75a}. \citet{Myer95b} strengthened Young's condition by demanding that for all profiles $R,R'$, there is some $k'\in\mathbb{N}$ such that for all $k\geq k'$, it holds that $f(kR+R') \subseteq f(R)$. The idea is that overwhelmingly large groups of voters should decide the outcome, while negligible, arbitrarily small groups should only be able to break ties. 
Note that our continuity condition is weaker than this one.

Finally, resolvability is a condition concerned with the decisiveness of a pSCF. It demands that whenever a pSCF returns more than one lottery, a single voter can be added such that the pSCF only returns a single lottery for the resulting profile. It is a weaker version of the \emph{single-voter resolvability} condition by \citet{Tide86a}, which requires that each of the tied winners can be made the unique winner by adding a single voter \citep[see also][]{Tide87a,HoPa20a}. Our version coincides with the one used by \citet{Holl24a}.

The key differences from the model by \citet{Bran13a} are replacing upper hemi-continuity with continuity, decisiveness with resolvability, and unanimity with faithfulness,\footnote{\citet{Bran13a} already remark that their results remain valid when unanimity is weakened to faithfulness.} adding homogeneity, and allowing for real-valued (rather than only rational-valued) probabilities.

Probabilistic social choice functions nest traditional (non-probabilistic) set-valued social choice functions.
If for each $A \in \fone(U)$ and $R\in\prof[A]$, $f(R) = \Delta(B)$ for some non-empty $B\subseteq A$, we say that $f$ is \emph{non-probabilistic}.
A non-probabilistic pSCF thus returns all lotteries over the alternatives selected by a set-valued social choice function, with which we identify it.
This identification behaves well: if a non-probabilistic pSCF satisfies one of our axioms, then the corresponding set-valued social choice function satisfies the analogous axiom in the standard (non-probabilistic) framework.

The central axioms we consider are population-consistency, cloning-consistency, composition-consistency, and Condorcet-consistency. They generalize the corresponding conditions for SCFs, with which they coincide on non-probabilistic pSCFs.

Population-consistency relates choices across varying electorates. More precisely, it requires that whenever a lottery is chosen simultaneously by two electorates, this lottery is also chosen by the union of both electorates. For example,
consider the two preference profiles $R'$ and $R''$ given below.
\[
\begin{array}{cc}
1 & 1 \\\midrule
a & b \\
b & c \\
c & a \\[2ex]
\multicolumn{2}{c}{R'}\\
\end{array}
\qquad\qquad
\begin{array}{cc}
1 & 1 \\\midrule
a & b \\
c & c \\
b & a \\[2ex]
\multicolumn{2}{c}{R''}\\
\end{array}
\qquad\qquad
\begin{array}{lcccr}
&1 & 1 & 2 & \\\cmidrule{2-4}
&a & a & b &\\
&b & c & c &\\
&c & b & a &\\[2ex]
\multicolumn{5}{c}{R' + R''}\\
\end{array}
\tag{Example 1}
\]
Population-consistency then demands that any lottery that is chosen in both $R'$ and $R''$ (say, $\nicefrac{1}{2}\, a + \nicefrac{1}{2}\, b$) also has to be chosen when both preference profiles are merged.
Formally, a pSCF satisfies population-consistency if for all $A\in\fone(U)$, $R',R''\in\prof[A]$,

\[f(R') \cap f(R'') \quad\subseteq\quad 
f(R' + R'')\text{.}
\tag{population-consistency}\]

Population-consistency is arguably one of the most natural axioms for variable electorates and is usually considered in a slightly stronger version, known as \emph{reinforcement} or simply \emph{consistency}, where the inclusion in the equation above is replaced with equality whenever the left-hand side is non-empty.
Note that population-consistency is merely a statement about abstract sets of outcomes, which makes no reference to lotteries whatsoever.
 It was first considered independently by \citet{Smit73a}, \citet{Youn74a}, and \citet{FiFi74a} and features prominently in the characterization of scoring rules by \citet{Smit73a} and \citet{Youn75a}. Population-consistency and its variants have found widespread acceptance in the social choice literature \citep[see, e.g.,][]{Youn74b,Fish78d,YoLe78a,Saar90a,Saar95a,Myer95b,CoMe12a,Piva13a,Bran13a,BrPe17a,BrPe19a,NePi18a,SFS19a,LaSk21a,Lede23a}. 

Cloning- and composition-consistency prescribe how pSCFs should deal with \emph{decomposable} preference profiles. 
For two agendas $A,B\in\fone(U)$, $B\subseteq A$ is a component in $R\in\prof[A]$ if the alternatives in $B$ are \emph{adjacent} in all preference relations that appear in $R$, i.e., for all $a\in A \setminus B$ and $b,b'\in B$, $a \pref b$ if and only if $a \pref b'$ for all ${\pref}\in\mathcal{L}(A)$ with $R({\pref})>0$. Intuitively, the alternatives in $B$ can be seen as variants or clones of the same alternative because they have exactly the same relationship to all alternatives that are not in $B$. For example, consider the following preference profile $R$ in which $B=\{b,b'\}$ constitutes a component. 

\[
\begin{array}{lll}
2 & 1 & 3 \\\midrule
a & a & b \\
b' & b & b' \\
b & b' & a \\[2ex]
\multicolumn{3}{c}{R}\\
\end{array}
\qquad\qquad
\begin{array}{ll}
3 & 3 \\\midrule
a & b \\
b & a \\ \\[2ex]
\multicolumn{2}{c}{R|_{A'}}\\
\end{array}
\qquad\qquad
\begin{array}{ll}
2 & 4 \\\midrule
b' & b \\
b & b' \\ \\[2ex]
\multicolumn{2}{c}{R|_B}\\ 
\end{array}
\tag{Example 2}
\]
The ``essence'' of $R$ is captured by $R|_{A'}$, where $A'=\{a,b\}$ contains only one of the cloned alternatives. It seems reasonable to demand that a pSCF should assign the same probability to $a$ (say, $\nicefrac{1}{2}$) independently of the number of clones of $b$ and the internal relationship between these clones. This condition is called cloning-consistency and was first proposed by \citet{Tide87a} \citep[see also][]{ZaTi89a}. 
For a formal definition of cloning-consistency, let $A',B\in\fone(U)$ and $A = A'\cup B$ such that $A'\cap B = \{b\}$. Then, a pSCF $f$ satisfies cloning-consistency if, for all $R\in\prof[A]$ such that $B$ is a component in $R$, 
\[\left\{(p_x)_{x\in A\setminus B}\colon p\in f(R)\right\} = \left\{(p_x)_{x\in A\setminus B}\colon p\in f(R|_{A'})\right\}\text{.}
\tag{cloning-consistency}
\]

On a second look at Example~2, it may appear strange that cloning-consistency does not impose any restrictions on how $f$ distributes probability among the clones. While clones behave completely identically with respect to uncloned alternatives, they are not indistinguishable from \emph{each other}. It seems that the relationships between clones ($R|_B$) should be taken into account as well. For example, one would expect that $f$ assigns more probability to $b$ than to $b'$ because two thirds of the voters prefer $b$ to $b'$. An elegant and mathematically appealing way to formalize this intuition is to require that the probabilities of the clones $b$ and $b'$ are directly proportional to the probabilities that $f$ assigns to these alternatives when restricting the preference profile to the component $\{b,b'\}$. This condition, known as \emph{composition-consistency}, is due to \citet{LLL96a} and was studied in detail for majoritarian SCFs \citep[see, e.g.,][]{Lasl96a,Lasl97a,Bran11b,BBS11a,Hora13a} and non-majoritarian SCFs \citep[see, e.g.,][]{Bran13a,OzTu20a,BCR+25a}.%

For a formal definition of composition-consistency, let $p\in\Delta(A')$ and $q\in\Delta(B)$ and define
	\[(p\times_b q)_x =
	\begin{cases}
		p_x &\text{if } x\in A\setminus B,\\
		p_b q_x\quad &\text{if } x \in B.
	\end{cases}\]
	The operator $\times_b$ is extended to sets of lotteries $X\subseteq\Delta(A')$ and $Y\subseteq\Delta(B)$ by applying it to all pairs of lotteries in $X\times Y$, i.e.,
	$X\times_b Y = \{p\times_b q\in\Delta(A)\colon p\in X \text{ and } q\in Y\}$.
	
	Then, a pSCF $f$ satisfies composition-consistency if for all $R\in\prof[A]$ such that $B$ is a component in $R$, 
\[f(R|_{A'})\times_b f(R|_B) \quad=\quad f(R) \text{.} \tag{composition-consistency}\]
In Example~2 above, if $\nicefrac{1}{2}\, a+\nicefrac{1}{2}\, b\in f(R|_{A'})$ and $\nicefrac{2}{3}\, b+\nicefrac{1}{3}\, b'\in f(R|_B)$, then composition-consistency implies that $\nicefrac{1}{2}\, a+\nicefrac{1}{3}\, b+\nicefrac{1}{6}\, b'\in f(R)$.

As pointed out by \citet{Bran13a}, composition-consistency implies cloning-consistency.

For $A\in \fone(U)$, $R \in \prof[A]$, $x,y \in A$, the entries $M_R(x,y)$ of the \emph{majority margin matrix} $M_R$ denote the difference between the number of voters who prefer~$x$ to~$y$ and the number of voters who prefer~$y$ to~$x$, i.e.,
\[M_R(x,y)= R(x,y) - R(y,x)\text.\]
An alternative $x \in A$ is a \emph{weak Condorcet winner} for $R$ if $M_R(x,y) \ge 0$ for all $y \in A$, and $x$ is a (strict) \emph{Condorcet winner} for $R$ if $M_R(x,y) > 0$ for all $y \in A\setminus\{x\}$.

A pSCF $f$ is Condorcet-consistent if for all $A\in \fone(U)$, $R \in \prof[A]$
\[ 
f(R)=\{x\}\text{ whenever $x$ is a Condorcet winner for $R$.}\tag{Condorcet-consistency}
\]
$f$ is \emph{weakly Condorcet-consistent} if $x\in f(R)$ whenever $x$ is a weak Condorcet winner for $R$. The continuity of pSCFs implies that every Condorcet-consistent pSCF is also weakly Condorcet-consistent. To see this, let $R \in \prof[A]$ be a profile with weak Condorcet winner $x$ and $R' \in \prof[A]$ a profile consisting of a single voter who ranks $x$ first. 
We then have $f(R' + kR) = \{x\}$ for all $k\in\mathbb N$ because $x$ is a Condorcet winner in all these profiles. Hence, continuity implies that $x \in f(R)$.

\section{Results}
\label{sec:results}

\citet[][Theorem~1]{Bran13a} have shown that population-consistency and cloning-consistency are incompatible for non-probabilistic social choice functions in a fractional profile framework.\footnote{They state their theorem for composition-consistency and mention in Footnote~11 that cloning-consistency suffices. Since their proof uses composition-consistency, we give a proof for cloning-consistency in \Cref{app:proof-non-probabilistic}.}
We transfer this statement to the framework of pSCFs via a mapping from one setting to the other (\Cref{lem:fractional-extension} in \Cref{app:prelims}).

\begin{theorem}
	\label{thm:non-probabilistic}
	There is no non-probabilistic pSCF that satisfies population-consistency and cloning-consistency.
\end{theorem}

It is instructive to compare this theorem to other results concerning the compatibility of two types of consistency conditions. 
Population-consistency and Condorcet-consistency are incompatible for set-valued social choice functions \citep{YoLe78a}. For three alternatives, this incompatibility already arises when only electorates with at most nine voters are considered, and this bound is tight \citep{BDP24a}.
Composition-consistency and Condorcet-consistency, on the other hand, can be satisfied simultaneously. This is, for example, the case for the \emph{essential set}, which returns all alternatives that receive positive probability in a maximal lottery, as well as for several variants of the \emph{uncovered set} and the \emph{Banks set}.
Some more discriminating set-valued social choice functions, such as \emph{ranked pairs} and \emph{split cycle}, are Condorcet-consistent and cloning-consistent.

Let us now turn to the main theorem, a characterization of a pSCF called maximal lotteries, proposed by \citet{Fish84a}.\footnote{\citet{Krew65a} proposed a majoritarian version of maximal lotteries much earlier than Fishburn. This version was independently rediscovered by \citet{FeMa92a}, \citet{LLL93b}, and \citet{FiRy95a}.} 
It was rediscovered and studied extensively in subsequent work \citep[see, e.g.,][]{DuLa99a,Lasl00a,RiSh10a,BBS20a,BrBr17a,BrBr21a}.\footnote{For overviews of the properties of maximal lotteries, the reader is referred to \citet{Bran17a}, \citet{BBS20a}, and \citet{Bran22a}.}
Maximal lotteries (\ml) can be viewed as ``probabilistic weak Condorcet winners.'' Formally, for all $A\in\fone(U)$ and $R\in\prof[A]$,
\[
\ml(R)=\{p\in \Delta(A) \colon 
p^\mathsf{T} M_R \ge \mathbf{0}\}\text.
\tag{maximal lotteries}
\]
$M_R$ can be interpreted as the payoff matrix of a symmetric zero-sum game and maximal lotteries as the mixed maximin strategies (or Nash equilibrium strategies) of this game. It thus follows from the minimax theorem that $\ml(R)\neq\emptyset$ for all $R\in\mathcal{R}$ \citep{vNeu28a}. Moreover,
$\lvert \ml(R)\rvert=1$ whenever the number of voters in $R$ is odd \citep{LLL97a}. This implies that $\ml$ not only satisfies resolvability, but an even stronger condition:
For all $A\in\fone(U)$ and $R \in \prof[A]$ such that $\lvert\ml(R)\rvert>1$, $\lvert \ml(R + {\pref})\rvert = 1$ for \emph{all} ${\pref} \in\mathcal L(A)$. This strong notion of resolvability is violated by virtually all common non-probabilistic pSCFs, including plurality and Borda's rule. 
The other three conditions we demand from pSCFs (homogeneity, faithfulness, and continuity) are easily seen to hold. Moreover, \citet[][Lemma~3]{Bran13a} have shown that \ml satisfies population-consistency and composition-consistency, and their proof carries over verbatim to our setting. The proof of the following theorem is deferred to \Cref{app:proof-main}.

\begin{theorem}\label{thm:main}
    $\ml$ is the only pSCF that satisfies population-consistency and composition-consistency.
\end{theorem}

\citet[][Remark~5]{Bran13a} also discuss an alternative characterization of \ml, which can be transferred to the framework of pSCFs. It weakens composition-consistency to cloning-consistency at the expense of adding weak Condorcet-consistency, and requires the pSCF to be rational-valued.
A pSCF $f$ is \emph{rational-valued} if $f(R)=\conv(f(R)\cap\mathbb Q^A)$ for all $A\in\fone(U)$ and $R\in\prof[A]$, i.e., $f(R)$ is the convex hull of the rational lotteries it contains. Since $\ml(R)$ is a polytope with rational vertices, $\ml$ is rational-valued.

\begin{theorem}\label{thm:alt}
    \ml is the only rational-valued pSCF that satisfies population-consistency, cloning-consistency, and weak Condorcet-consistency.
\end{theorem}
Note that \ml also satisfies composition-consistency (which is stronger than cloning-consistency) and Condorcet-consistency (which is stronger than weak Condorcet-consistency). \Cref{thm:alt} does not require that pSCFs satisfy faithfulness.

Convex-valuedness cannot be dropped from the definition of pSCFs in \Cref{thm:main}%
: the function that maps each preference profile $R$ to the set of all maximal lotteries with \emph{rational} probabilities, $\ml(R)\cap\mathbb{Q}^A$, satisfies all remaining axioms, including population-consistency, composition-consistency, and Condorcet-consistency, but differs from $\ml$ whenever $\ml(R)$ is not a singleton.
Whether convex-valuedness is required was left open by \citet[][Footnote~18]{Bran13a}. The example above settles this question in our framework; in the framework of \citeauthor{Bran13a}, where all lotteries are rational-valued, it remains open.

Faithfulness and resolvability cannot be dropped from \Cref{thm:main} either.
The function that returns $\Delta(A)$ for every $A\in\fone(U)$ and $R\in\prof[A]$ satisfies all our conditions except resolvability, as well as population-consistency and composition-consistency.\footnote{Composition-consistency holds because $\Delta(A')\times_b\Delta(B) = \Delta(A)$ whenever $A = A'\cup B$ and $A'\cap B=\{b\}$.}
Since this function is non-probabilistic, resolvability cannot be dropped from \Cref{thm:non-probabilistic} either.
For $R\in\prof[A]$, let $\bar R$ denote the profile obtained from $R$ by inverting the preference relation of every voter, so that $M_{\bar R} = -M_R$.
The function that maps $R$ to $\ml(\bar R)$ satisfies all our conditions except faithfulness, as well as population-consistency and composition-consistency, but returns $\{y\}$ when a single voter prefers $x$ to $y$.
The function that returns the relative interior of \ml satisfies all axioms except continuity.
Neither population-consistency nor composition-consistency can be dropped: random dictatorship satisfies all our conditions and population-consistency, but violates composition-consistency, whereas the variant of \ml in which all majority margins are replaced by their third powers satisfies all our conditions and composition-consistency, but violates population-consistency \citep[][Remark~1]{Bran13a}.
Whether homogeneity can be dropped from \Cref{thm:main} remains open.
It is only needed for the inclusion $f\subseteq\ml$: every function that satisfies all remaining conditions and $f\subseteq\ml$ also satisfies $\ml\subseteq f$, since the proof of \Cref{lem:ml-subseteq-f} can be carried out directly for profiles with finitely many voters.%

\Cref{thm:main} requires the universe of alternatives $U$ to be infinite, and this assumption cannot be replaced by a finite lower bound on $|U|$: for finite universes, both inclusions $f\subseteq\ml$ and $\ml\subseteq f$ may fail.
For every $|U|\geq 5$, there is a pSCF satisfying population-consistency and composition-consistency that returns all maximal lotteries and, at some profiles, additional non-maximal lotteries with irrational probabilities.
For every $|U|\geq 4$, there is such a pSCF that returns a proper subset of $\ml(R)$ at some profiles $R$.\footnote{The first construction exploits that we admit real-valued probabilities, and the second that our continuity condition is weaker than upper hemi-continuity.
Neither construction thus applies to the original framework of \citet{Bran13a}, which assumes upper hemi-continuity and rational-valued probabilities; whether finite universes suffice there remains open.}
Nevertheless, it can be shown that the inclusion $f\subseteq\ml$ holds approximately for finite universes.

\section{Related Game-Theoretic Characterizations}

\Cref{thm:alt} bears some resemblance to the characterization of the Nash equilibrium correspondence by \citet{BrBr23a}.
They consider total solution concepts that return a non-empty set of mixed strategy profiles for every finite normal-form game with a fixed set of players and show that the Nash equilibrium correspondence is the unique solution concept satisfying consistency, consequentialism, and rationality.
Consistency requires that every strategy profile returned in both of two games with the same action sets is also returned in every convex combination of them.
This is analogous to population-consistency: after passing to the homogeneous fractional-profile representation, adding electorates corresponds to taking rational convex combinations of profiles, and hence of the induced majority margin matrices.
Consequentialism requires invariance under blowing up actions into payoff-equivalent clones and, in particular, requires that probability can be shifted arbitrarily between clones.
This resembles cloning-consistency, but the analogy is not exact.
A component of a preference profile induces identical majority margins against all alternatives outside the component, yet the alternatives in the component need not be payoff-equivalent in the induced symmetric zero-sum game, because their majority margins against each other may be non-zero.
Conversely, alternatives that are payoff-equivalent in the majority margin matrix need not form a component of the underlying preference profile.
Rationality has no exact counterpart in \Cref{thm:alt}; the closest counterpart is weak Condorcet-consistency, which requires every weak Condorcet winner---a pure Nash equilibrium strategy---to be included in the choice set, whereas the rationality axiom of \citet{BrBr23a} only requires strictly dominant pure actions to receive positive probability.

Thus, even though maximal lotteries are the maximin strategies (equivalently, the Nash equilibrium strategies) of the symmetric zero-sum game induced by $M_R$, \Cref{thm:alt} is logically independent of the Nash equilibrium characterization.
It is formulated for preference profiles, not for arbitrary normal-form games: the majority margin map yields only symmetric zero-sum games, a strict subclass on which the uniqueness of the Nash correspondence over all games has no bearing.
Moreover, pSCFs return sets of lotteries rather than sets of strategy profiles, and consequentialism has no exact counterpart: cloning-consistency only requires that the probabilities of the remaining alternatives are unaffected by cloning, composition-consistency prescribes how probability is split among clones, and neither requires that every split is chosen.
The maximin characterization of \citet{BrBr17c} is analogous in the same sense: it uses consequentialism, consistency, and a dominance-based rationality axiom to characterize maximin strategies in two-player zero-sum games.
It is closer to \Cref{thm:alt} because maximal lotteries are precisely maximin strategies of the induced symmetric zero-sum games, although the same domain and cloning caveats apply.
The theorems by \citet{BrBr17c,BrBr23a} require neither continuity nor convex-valuedness nor decisiveness because their axioms rule out the game-theoretic analogues of the examples given after \Cref{thm:alt}.
Consequentialism demands that probability can be shifted arbitrarily between clones, which is violated by the analogues of the relative interior of $\ml$ (which never assigns probability~$0$ to a clone of an action in its support) and of $R\mapsto\ml(R)\cap\mathbb Q^A$ (which excludes irrational splits).
Both functions satisfy composition-consistency since it prescribes how probability is split among clones.
Rationality is violated by the solution concept that returns all strategy profiles, the analogue of $R\mapsto\Delta(A)$, whereas faithfulness is too weak to rule out the latter.

\subsection*{Statements and Declarations}
This material is based on work supported by the Deutsche Forschungsgemeinschaft under grants {BR~2312/14-1} and {EXC-2047}.
The authors have no relevant financial or non-financial interests to disclose.

\newpage

\appendix

\section*{APPENDIX}

\section{Preliminaries}\label{app:prelims}

Throughout this appendix, we exploit the connection between our model and the fractional-profile framework of \citet{Bran13a}. Since pSCFs are homogeneous, every pSCF~$f$ induces a PSCF~$F$ on fractional profiles (\Cref{lem:induced-PSCF}) that inherits the consistency properties of~$f$ (\Cref{lem:fractional-extension}); \Cref{thm:main} thus reduces to showing that $F = \ml$, which lets us reuse much of the machinery of \citet{Bran13a}. 
We establish the two inclusions separately. For $F\subseteq\ml$ (\Cref{lem:f-subseteq-ml}), we reprove the two-alternative case (\Cref{lem:two-alternatives}) from scratch, since our faithfulness and continuity axioms are weaker than the unanimity and upper hemi-continuity assumed by \citeauthor{Bran13a}; for larger agendas, their argument carries over, except that, as we admit real-valued (rather than only rational-valued) probabilities, we first reduce any violation of maximality to one with rational probabilities (\Cref{lem:rational-violation}). This reduction invokes the fact that population- and composition-consistency imply weak Condorcet-consistency near the uniform profile, which we restate as \Cref{lem:condorcet-consistency} \citep[][Lemma~6]{Bran13a}. Apart from the two-alternative case, where it ensures that both alternatives are chosen when they are tied, continuity is only used for the reverse inclusion $\ml\subseteq F$ (\Cref{lem:ml-subseteq-f}). Because we weaken upper hemi-continuity, the argument of \citeauthor{Bran13a} (their Lemma~15) no longer applies, and we give a self-contained proof in which the relevant maximal lottery is the \emph{unique} one along a sequence of profiles converging to the given profile. It relies on a construction of skew-symmetric matrices that we recall as \Cref{lem:mcgarvey} \citep[][Lemma~14]{Bran13a}; along the way, we 
simplify the argument by reducing the case of an even-sized support to that of an odd-sized one.\footnote{Some typos and minor issues in the appendix by \citet{Bran13a} have been corrected in an updated online version of their paper \citep{BBS15a}.}

For a finite set $X$, let $\Delta_\mathbb{Q}(X) = \Delta(X) \cap\mathbb Q^{X}$ be the set of probability distributions with rational values.
The set of fractional preference profiles for an agenda $A\in\fone(U)$ is $\fprof[A] = \Delta_{\mathbb Q}(\mathcal L(A))$, which can be associated with the $(|A|!-1)$-dimensional unit simplex in $\mathbb Q^{\mathcal L(A)}$. 
Therefore, $R({\pref})$ is the fraction of voters with preference relation ${\pref}\in\mathcal{L}(A)$. 
The set of fractional preference profiles is $\mathcal R^* = \bigcup_{A\in\fone(U)} \fprof[A]$.
As in the main text, we write ${\pref}$ for the profile $R\in\fprof[A]$ with $R({\pref})=1$, in which all voters have preference relation ${\pref}$.

For a finite set $X$, $\uni(X)$ denotes the uniform distribution over $X$; in particular, $\uni(A)$ is the uniform lottery over $A$ and $\uni(\mathcal L(A))$ is the uniform profile on $A$.
The \emph{support} of a lottery $p\in\Delta(A)$ is $\supp(p)=\{x\in A\colon p_x>0\}$.
We also write $p(x)$ for $p_x$ and let $p(X)=\sum_{x\in X}p(x)$ for $X\subseteq A$.
For $n\in\mathbb N$, let $[n]=\{1,\dots,n\}$, and for a finite set $X$, let $\Pi(X)$ be the set of all permutations of $X$.
Finally, $\lVert v\rVert=\sum_{i}|v_i|$ denotes the $1$-norm of a vector $v$.

A probabilistic social choice function $F$ on fractional profiles (PSCF) is a function that, for any agenda $A\in\fone(U)$, maps a fractional preference profile $R\in \fprof[A]$ to a non-empty and convex subset of $\Delta(A)$, and satisfies the following three properties.
\settowidth{\firstcol}{\emph{Homogeneity:}}

\medskip
\noindent
\begin{tabular}{@{}p{\firstcol}p{\linewidth-\firstcol-\tabcolsep-\tabcolsep}@{}} 
\emph{Faithfulness:} & For all $x,y\in U$,
$
    R(\{(x,y)\}) = 1 \text{ implies } F(R) \neq \{y\}
$.\\
\emph{Continuity:} & For all $A\in\fone(U)$, $R,R'\in\fprof[A]$, and $p\in\Delta(A)$, if
$F(\nicefrac1k\, R' + (1 - \nicefrac1k)\, R) = \{p\}$ for all $k\in\mathbb N$, then $p \in F(R)$.\\
\emph{Decisiveness:} & For all $A\in\fone(U)$, $\{R\in\fprof[A]\colon |F(R)| = 1\}$ is dense in $\fprof[A]$.\\
\end{tabular}
\medskip

Note that the continuity condition used here is weaker than the upper hemi-continuity condition used by \citet{Bran13a}. As a consequence, we need to reprove some of their statements using this weaker condition. 

A PSCF is non-probabilistic if for each $A\in\fone(U)$ and $R\in\fprof[A]$, $F(R) = \Delta(B)$ for some non-empty $B\subseteq A$.
A PSCF satisfies population-consistency if for all $A\in\fone(U)$, $R',R''\in\fprof[A]$, and $\lambda\in[0,1]\cap\mathbb Q$,
\[F(R') \cap F(R'') \quad\subseteq\quad
F(\lambda\, R' + (1-\lambda)\, R'')\text{.}
\tag{population-consistency}\]

A PSCF satisfies composition-consistency, if for all $A',B\in\fone(U)$ and $A = A'\cup B$ such that $A'\cap B = \{b\}$, and all $R\in\fprof[A]$ such that $B$ is a component in $R$, 
\[F(R|_{A'})\times_b F(R|_B) \quad=\quad F(R) \text{.} \tag{composition-consistency}\]
Moreover, it satisfies cloning-consistency if for all $A',B\in\fone(U)$ and $A = A'\cup B$ such that $A'\cap B = \{b\}$, and all $R\in\fprof[A]$ such that $B$ is a component in $R$, 
\[\{(p_x)_{x \in A\setminus B}\colon p \in F(R)\} = \{(p_x)_{x\in A\setminus B}\colon p \in F(R|_{A'})\} \text{.} \tag{cloning-consistency}\]
Composition-consistency implies cloning-consistency.
Majority margins, (weak) Condorcet winners, (weak) Condorcet-consistency, and rational-valuedness are defined for PSCFs in the same way as for pSCFs, where $R(x,y)$ now denotes the fraction of voters who prefer $x$ to $y$.
A PSCF $F$ is \emph{weakly Condorcet-consistent at} $R\in\fprof$ if $x\in F(R)$ for every weak Condorcet winner $x$ for $R$.

Let $f$ be a pSCF. The function $F$ \emph{induced} by $f$ is defined as follows: for each $A\in\fone(U)$ and $R\in\fprof[A]$, let $n\in\mathbb N$ be such that $nR \in \prof[A]$ and let
\[F(R) = f(nR).\]
\begin{lemma}
	\label{lem:induced-PSCF}
	Let $f$ be a pSCF. Then, the function $F$ induced by $f$ is a well-defined PSCF.
\end{lemma}
\begin{proof}
  First, $F$ is well-defined.
  Let $R\in\fprof[A]$ and let $m,n\in\mathbb N$ such that $mR,nR\in\prof[A]$.
  By homogeneity of $f$, $f(mR)=f(nmR)=f(mnR)=f(nR)$, and thus $F(R)$ does not depend on the chosen scaling factor.

  Second, $F$ is a PSCF.
  Faithfulness is immediate: if $R(\{(x,y)\}) = 1$, then $R = 1R$ is a non-fractional profile with one voter who prefers $x$ to $y$, and therefore $F(R) = f(R) \neq \{y\}$.
  Convex-valuedness holds because $F(R)=f(nR)$ for some $n\in \mathbb N$ and $f(nR)$ is convex.
  For continuity, fix $R\in\fprof[A]$.
  Suppose there is $R'\in\fprof[A]$ such that $F(\nicefrac1k\, R' + (1-\nicefrac1k)\, R) = \{p\}$ for all $k\in\mathbb N$.
  Let $n\in\mathbb N$ such that $nR\in\prof[A]$ and $nR'\in\prof[A]$.
  Then, $f(nR' + (k-1)nR) = \{p\}$ for all $k\in\mathbb N$.
  Because $f$ is continuous, $p \in f(nR) = F(R)$.
  For decisiveness, fix $R\in\fprof[A]$.
  If $|F(R)|=1$, then $R$ itself is resolute, so assume $|F(R)|>1$ and choose $n\in\mathbb N$ with $nR\in\prof[A]$.
  By homogeneity, $|f(knR)|=|f(nR)|=|F(R)|>1$ for every $k\in\mathbb N$, so, for each $k\in\mathbb N$, resolvability of $f$ yields ${\pref}_k\in\mathcal L(A)$ such that
  $|f(knR+{\pref}_k)|=1$.
  Let
  \[
    R_k=\frac{knR+{\pref}_k}{kn+1}=\frac{kn}{kn+1}R+\frac{1}{kn+1}{\pref}_k\text{.}
  \]
  Then $R_k\to R$ as $k\to\infty$ and
  $F(R_k)=f(knR+{\pref}_k)$, so $|F(R_k)|=1$.
  Thus $\{R\in\fprof[A]\colon |F(R)|=1\}$ is dense in $\fprof[A]$.	
\end{proof}
The PSCF induced by the pSCF $\ml$ is also denoted by $\ml$.
A pSCF that satisfies any of our axioms induces a PSCF that also satisfies that axiom.
\begin{lemma}\label{lem:fractional-extension}
  	Let $f$ be a pSCF and denote by $F$ the PSCF induced by $f$.
	If $f$ is population-consistent, composition-consistent, cloning-consistent, (weakly) Condorcet-consistent, or non-probabilistic, then so is $F$.
\end{lemma}
\begin{proof} 
	To prove population-consistency, let $R',R''\in\fprof[A]$, $\lambda\in[0,1]\cap\mathbb Q$, and $p\in F(R')\cap F(R'')$.
  The cases $\lambda\in\{0,1\}$ are trivial, so assume $\lambda\in(0,1)$.
  Choose $m\in\mathbb N$ such that $m\lambda,m(1-\lambda)\in\mathbb N$, $m\lambda R'\in\prof[A]$, and $m(1-\lambda)R''\in\prof[A]$.
  By homogeneity, $F(R')=f(m\lambda R')$ and $F(R'')=f(m(1-\lambda)R'')$.
  Population-consistency of $f$ yields $p\in f(m\lambda R'+m(1-\lambda)R'')$, and since
  $m(\lambda R'+(1-\lambda)R'')=m\lambda R'+m(1-\lambda)R''$, we obtain
  $p\in F(\lambda R'+(1-\lambda)R'')$.

  To prove composition-consistency, let $A',B\in\fone(U)$ with $A'\cap B=\{b\}$, let $A=A'\cup B$, and let
  $R\in\fprof[A]$ such that $B$ is a component in $R$.
  Choose $m$ with $mR\in\prof[A]$.
  Then $B$ is also a component in $mR$, and composition-consistency of $f$ gives
  \[
    f(mR)=f(mR|_{A'})\times_b f(mR|_B)\text{.}
  \]
  By homogeneity, $f(mR|_{A'})=F(R|_{A'})$ and $f(mR|_B)=F(R|_B)$, so
  $F(R)=F(R|_{A'})\times_b F(R|_B)$.
  
  The proof of cloning-consistency is similar.
  
  It is immediate that if $f$ is (weakly) Condorcet-consistent or non-probabilistic, then so is $F$.
  The proofs are omitted.
\end{proof}

The converse of \Cref{lem:induced-PSCF} fails: not every PSCF is induced by a pSCF.
For $A\in\fone(U)$ and $R\in\fprof[A]$, let $d(R)$ be the smallest $n\in\mathbb N$ such that $nR\in\prof[A]$, and let
\[
  F(R) = \begin{cases}
    \{\uni(A)\} & \text{if $3$ divides $d(R)$,}\\
    \Delta(A) & \text{otherwise.}
  \end{cases}
\]
Then $F$ is a PSCF: it is non-empty- and convex-valued; it is faithful because $d(R) = 1$ if $R(\{(x,y)\}) = 1$; it is continuous because whenever the antecedent applies, the lottery in question is $\uni(A)$, which is contained in $F(R)$ in either case; and it is decisive because the profiles $R$ with $3 \mid d(R)$ are dense in $\fprof[A]$.
Now suppose that $F$ were induced by some pSCF $f$. Then $f(R) = F(\nicefrac1n\, R)$ for every $A$ with $|A|\ge 2$ and every $R\in\prof[A]$ with $n$ voters.
Let $n = 3m+1$. Since $d(\nicefrac1n\, R)$ divides $n$ and $3\nmid n$, we have $f(R) = \Delta(A)$ and hence $\lvert f(R)\rvert > 1$.
Adding a single voter yields a profile with $3m+2$ voters, whose associated denominator divides $3m+2$ and is therefore not divisible by $3$ either.
Hence $\lvert f(R+{\pref})\rvert > 1$ for all ${\pref}\in\mathcal L(A)$, i.e., $f$ violates resolvability.
The reason is that decisiveness only requires the resolute profiles to be dense, whereas resolvability requires them to be reachable by adding a single voter.

In the proofs below, we will make use of the fact that cloning-consistency implies \emph{neutrality}, i.e., $\pi(F(R)) = F(\pi(R))$ for all $A\in\fone(U)$, $R\in\fprof[A]$, and all permutations $\pi$ of $A$, where $\pi(R)$ and $\pi(F(R))$ result from relabeling the alternatives according to $\pi$. While \citet[][Lemma 1]{Bran13a} state this implication for composition-consistency, their proof only uses cloning-consistency. 

\section{Proof of \Cref{thm:non-probabilistic}}\label{app:proof-non-probabilistic}

\begin{proof}[Proof of \Cref{thm:non-probabilistic}]
Assume for contradiction that $f$ is a non-probabilistic pSCF that satisfies population-consistency and cloning-consistency, and let $F$ be the PSCF induced by $f$.
By \Cref{lem:fractional-extension}, $F$ is non-probabilistic and satisfies population-consistency and cloning-consistency, and hence neutrality.
For $R\in\fprof$, let $S(R)$ be the non-empty set of alternatives with $F(R)=\Delta(S(R))$.
Population-consistency implies that, for each alternative~$x$, the set of profiles $R$ with $x\in S(R)$ is closed under rational convex combinations.

First, let $x,y\in U$ be distinct and $R\in\fprof[\{x,y\}]$.
If $R(x,y)=R(y,x)$, then $S(R)=\{x,y\}$ by neutrality.
Now let $R(x,y)>R(y,x)$ and assume for contradiction that $y\in S(R)$.
Then $y\in S(R')$ for all $R'\in\fprof[\{x,y\}]$ with $R(x,y)\ge R'(x,y)\ge\nicefrac12$.
For the profile $R''$ in which all voters prefer $x$ to $y$, faithfulness gives $S(R'')\neq\{y\}$, so $x\in S(R'')$, and therefore $x\in S(R')$ for all $R'$ with $1\ge R'(x,y)\ge\nicefrac12$.
Hence, $S(R')=\{x,y\}$ for all $R'$ with $R(x,y)\ge R'(x,y)\ge\nicefrac12$, which contradicts decisiveness.
Thus, $S(R)=\{x\}$ whenever $R(x,y)>R(y,x)$.

Now let $A=\{a,b,c\}\in\fone(U)$, let $R\in\fprof[A]$, and let $x,y,z$ be such that $A=\{x,y,z\}$ and $\{y,z\}$ is a component in $R$.
Applying cloning-consistency with $A'=\{x,y\}$ and $B=\{y,z\}$ and using the two-alternative case, we obtain:
\begin{enumerate}[label=\textit{(\alph*)}]
	\item if $R(x,y)>R(y,x)$, then $S(R)=\{x\}$; \label{item:np-win}
	\item if $R(x,y)=R(y,x)$, then $x\in S(R)$ and $S(R)\cap\{y,z\}\neq\emptyset$; \label{item:np-tie}
	\item if $R(x,y)<R(y,x)$, then $x\notin S(R)$. \label{item:np-lose}
\end{enumerate}
We write $abc$ for the profile in which all voters have the preference relation $a\succ b\succ c$, and similarly for the other five preference relations over $A$.
Consider the profiles
\begin{align*}
	R^1 &= \tfrac13(abc+bca+cab)\text, & R^2 &= \tfrac13(acb+cba+bac)\text,\\
	R^3 &= \tfrac12(acb+bca)\text, & R^4 &= \tfrac12(abc+bac)\text,\\
	R^5 &= \tfrac14(acb+cab+bca+bac)\text, & R^6 &= \tfrac1{13}(3\,abc+2\,acb+2\,bac+2\,bca+2\,cab+2\,cba)\text.
\end{align*}
We show that $\{a,b\}\subseteq S(R^i)$ for all $i\in\{1,\dots,6\}$.
Since $R^1$ and $R^2$ are invariant under the cyclic permutation $a\mapsto b\mapsto c\mapsto a$, neutrality implies $S(R^1)=S(R^2)=A$.
In $R^3$, both $\{b,c\}$ and $\{a,c\}$ are components and $a$ and $b$ are tied, so \ref{item:np-tie} yields $a,b\in S(R^3)$.
In $R^4$, $\{a,b\}$ is a component and all voters prefer $a$ to $c$, so $c\notin S(R^4)$ by \ref{item:np-lose}; since $R^4$ is invariant under swapping $a$ and $b$, neutrality implies $S(R^4)=\{a,b\}$.
In $R^5$, $\{a,c\}$ is a component and $a$ and $b$ are tied, so $b\in S(R^5)$ and $S(R^5)\cap\{a,c\}\neq\emptyset$ by \ref{item:np-tie}; since $R^5$ is invariant under swapping $a$ and $c$, we get $S(R^5)=A$.
Finally, $R^6=\tfrac6{13}R^1+\tfrac6{13}R^2+\tfrac1{13}abc$, and $S(abc)=\{a\}$ by \ref{item:np-win} since $\{b,c\}$ is a component in $abc$; hence, $a\in S(R^6)$.
Moreover, $R^6=\tfrac3{13}R^1+\tfrac2{13}T^1+\tfrac4{13}T^2+\tfrac4{13}T^3$, where $T^1=\tfrac12(acb+bac)$, $T^2=\tfrac14(abc+acb+bca+cba)$, and $T^3=\tfrac14(abc+bac+cab+cba)$.
In $T^1$, $\{a,c\}$ is a component and $a$ and $b$ are tied, so $b\in S(T^1)$ by \ref{item:np-tie}.
In $T^2$ and $T^3$, the components $\{b,c\}$ and $\{a,b\}$, respectively, are tied with the remaining alternative, and the profiles are invariant under swapping the two alternatives in the component; hence, $S(T^2)=S(T^3)=A$ by \ref{item:np-tie} and neutrality.
Thus, $b\in S(R^6)$.

The profiles $R^1,\dots,R^6$ are affinely independent, so their convex hull has non-empty interior in $\fprof[A]$.
By population-consistency, $\{a,b\}\subseteq S(R)$ for every rational convex combination $R$ of $R^1,\dots,R^6$, which contradicts decisiveness.
\end{proof}

\section{Proof of \Cref{thm:main}}\label{app:proof-main}
\subsection{$F\subseteq \ml$}

First, we consider the case of two alternatives.

\begin{lemma}\label{lem:two-alternatives}
    Let $F$ be a PSCF that satisfies population-consistency and composition-consistency.
    Let $x,y\in U$ be distinct.
    Then, for each $R\in\fprof[\{x,y\}]$, $F(R) = \ml(R)$.
\end{lemma}
\begin{proof}
  Let $R\in \fprof[\{x,y\}]$ with $R(x,y) > R(y,x)$.
  \setcounter{step}{0}
  \begin{step}\label{step:two-alternatives1}
    First, we show that for each $p\in F(R)$, $p(x) \in\{0,1\}$, that is, $F(R) \subseteq\{x,y\}$.
  Assume for contradiction that there is $p\in F(R)$ with $p(x) \in (0,1)$.
  Composition-consistency and convex-valuedness imply that $\relint(\Delta(\{x,y\}))\subseteq F(R)$; this can be extracted from the proof of Lemma~4 of \citet{Bran13a}, which uses continuity only to obtain the two degenerate lotteries.
Since $F$ satisfies composition-consistency, it also satisfies neutrality, which implies that $\relint(\Delta(\{x,y\}))\subseteq F(R^{x\leftrightarrow y})$, where $R^{x\leftrightarrow y}$ is the profile obtained from $R$ by swapping the labels of $x$ and $y$.
Then, population-consistency implies that $\relint(\Delta(\{x,y\}))\subseteq F(R')$ for each $R'\in\fprof[\{x,y\}]$ with $R(x,y) \ge R'(x,y) \ge R(y,x) = R^{x\leftrightarrow y}(x,y)$.
This contradicts decisiveness.
  \end{step}
  \begin{step}
    Second, we show that $y\not\in F(R)$.
Assume for contradiction that $y\in F(R)$.
If $R'\in\fprof[\{x,y\}]$ satisfies $R'(x,y) = 0$, then all voters prefer $y$ to $x$; by \Cref{step:two-alternatives1} (applied to $y$ via neutrality) and convex-valuedness, $F(R')$ is a single degenerate lottery, and since faithfulness yields $F(R')\neq\{x\}$, we obtain $y \in F(R')$.
Population-consistency implies that $y\in F(R')$ for each $R'\in \fprof[\{x,y\}]$ with $R(x,y) \ge R'(x,y) \ge 0$.
By neutrality, $x\in F(R')$ for each $R'\in\fprof[\{x,y\}]$ with $1 \ge R'(x,y) \ge R(y,x)$.
Thus, by convex-valuedness, $\Delta(\{x,y\})\subseteq F(R')$ for each $R'\in\fprof[\{x,y\}]$ with $R(x,y) \ge R'(x,y) \ge R(y,x)$.
This contradicts \Cref{step:two-alternatives1}.
Together, it follows that $F(R) = \{x\}$.    
  \end{step}
Neutrality implies that $F(R') = \{y\}$ for each $R'\in\fprof[\{x,y\}]$ with $R'(y,x) > R'(x,y)$.
It follows that $F(R')\subseteq \ml(R')$ for each $R'\in\fprof[\{x,y\}]$.
\begin{step}
  Third, we show that $F(R') = \ml(R')$ if $R'(x,y) = R'(y,x)$.
Let $\pref\in\mathcal L(\{x,y\})$ with $x\pref y$.
By the previous two steps, for each $\lambda\in(0,1]\cap\mathbb Q$, $F(\lambda\pref + (1-\lambda) R') = \{x\}$.
Thus, continuity implies that $x \in F(R')$.
By neutrality, $y \in F(R')$.
Thus, convex-valuedness implies that $F(R') = \Delta(\{x,y\})$.
\end{step}
\end{proof}

\begin{lemma}[\citealp{Bran13a}, Lemma 6]
	\label{lem:condorcet-consistency}
	Let $F$ be a PSCF that satisfies population-consistency and composition-consistency.
	Then, for each $A\in\fone(U)$, $F$ satisfies weak Condorcet-consistency in a neighborhood of the uniform profile $\uni(\mathcal L(A))$.
\end{lemma}
The proof by \citet{Bran13a} carries over verbatim, with \Cref{lem:two-alternatives} in place of their Lemma~5.

\begin{lemma}\label{lem:rational-violation}
  Let $F$ be a PSCF that satisfies population-consistency and composition-consistency.
  If $F\not\subseteq \ml$, then there are $A\in\fone(U)$, $R\in\fprof[A]$, and $p\in F(R)\setminus\ml(R)$ with $p\in\Delta_{\mathbb Q}(A)$.
\end{lemma}

\begin{proof}
We first record two constructions that will each be used twice.

First, let $X,Y\in\fone(U)$, $S\in\fprof[X]$, $q\in F(S)$, and $a\in X$ with
$Y\cap X=\{a\}$, and let $S'$ be the profile on $(X\setminus\{a\})\cup Y$ obtained
from $S$ by replacing $a$ with the component $Y$ with internal profile
$\uni(\mathcal L(Y))$.
Since every alternative in $Y$ is a weak Condorcet winner for $\uni(\mathcal L(Y))$,
\Cref{lem:condorcet-consistency} and convex-valuedness imply that
$F(\uni(\mathcal L(Y)))=\Delta(Y)$.
Hence, composition-consistency yields $q\times_a t\in F(S')$ for every
$t\in\Delta(Y)$.
In $S'$, every alternative in $Y$ has the same majority margins against alternatives
outside of $Y$ as $a$ has in $S$, the majority margins within $Y$ are $0$, and all
other majority margins are unchanged.
We refer to this construction as \emph{splitting} $a$ into $Y$.

Second, independently, let $X\in\fone(U)$, $T\in\fprof[X]$, $q\in F(T)$, and
$\Phi\subseteq X$ such that $q$ is constant on $\Phi$, and let $\Pi_\Phi$ contain
all permutations of $X$ that map
$\Phi$ onto itself and fix all other alternatives.
For each $\pi\in\Pi_\Phi$, neutrality implies that $q=\pi(q)\in F(\pi(T))$, and
population-consistency (applied repeatedly) yields $q\in F(T_\Phi)$, where
$T_\Phi=|\Pi_\Phi|^{-1}\sum_{\pi\in\Pi_\Phi}\pi(T)$.
In $T_\Phi$, each $u\in\Phi$ has majority margin
$|\Phi|^{-1}\sum_{v\in\Phi}M_T(v,z)$ against each $z\in X\setminus\Phi$, the
majority margins within $\Phi$ are $0$ by the skew-symmetry of $M_T$, and the
majority margins among $X\setminus\Phi$ are unchanged.
We refer to this construction as \emph{symmetrizing} $\Phi$.

We now prove the following claim.
\begin{claim}\label{claim:rational-violation}
There exist $A\in\fone(U)$, disjoint non-empty $W,B\subseteq A$,
$\tilde R \in\fprof[A]$, and $\tilde p\in F(\tilde R)$ such that
  \begin{enumerate}[label=\textit{(\roman*)}]
    \item $M_{\tilde R}(w,w')=0$ and $M_{\tilde R}(w,z)>0$ for all $w,w'\in W$ and
          $z\in A\setminus W$,\label{item:rational-violation1}
    \item $\tilde p$ is constant on $W$ and $\tilde p(W)<1$, and
          \label{item:rational-violation2}
    \item for all $x,y\in B$, $\tilde p(x)=\tilde p(y)>0$.
          \label{item:rational-violation3}
  \end{enumerate}
\end{claim}
Note that \ref{item:rational-violation1} implies $\ml(\tilde R)=\Delta(W)$:
for every $p$ with $\supp(p)\subseteq W$, the $W$-entries of
$p^\mathsf{T}M_{\tilde R}$ are $0$ and the remaining entries are positive, whereas
for every $p$ with $p(A\setminus W)>0$, every $W$-entry of
$p^\mathsf{T}M_{\tilde R}$ is negative.

\begin{proof}[Proof of \Cref{claim:rational-violation}]
Let $A_0\in\fone(U)$, $R_0\in\fprof[A_0]$, and $p_0\in F(R_0)\setminus\ml(R_0)$.
Then, there is $c\in A_0$ such that
\[
  \gamma=\sum_{a\in A_0}p_0(a)M_{R_0}(c,a)>0\text.
\]
In particular, $p_0(c)<1$, and so $D=\supp(p_0)\setminus\{c\}\neq\emptyset$.
Since $M_{R_0}(c,c)=0$, we have
$\sum_{a\in D}p_0(a)M_{R_0}(c,a)=\gamma$, and so $p_0|_D$ is contained in the open
cone $\{\nu\in\mathbb R^{D}_{>0}\colon \sum_{a\in D}\nu(a)M_{R_0}(c,a)>0\}$.
Hence, there are $\nu_1,\dots,\nu_N\in\mathbb Z_{\geq 0}^{D}$ with
$\sum_{a\in D}\nu_i(a)M_{R_0}(c,a)>0$ for each $i$ and
$\mu_1,\dots,\mu_N\in\mathbb R_{>0}$ with
$\sum_{i=1}^N\mu_i\nu_i=p_0|_D$.\footnote{Choose rational vectors in this cone that span
a simplicial cone containing $p_0|_D$ in its interior and scale them to integer
vectors.}
Moreover, choose $m\in\mathbb N$ such that
$m\sum_{a\in D}\nu_1(a)M_{R_0}(c,a)>1$.

Successively split each $a\in D$ into a block consisting of $a$ and fresh clones,
namely $\nu_i(a)$ \emph{mass clones} for each $i\in\{1,\dots,N\}$ and
$m\,\nu_1(a)$ \emph{padding clones} for each $a'\in A_0\setminus\{c\}$, and
distribute the probability $p_0(a)$ such that each mass clone corresponding to $i$
receives $\mu_i$ while $a$ itself and all padding clones receive $0$; this is
possible since
$\sum_{i=1}^N\mu_i\nu_i(a)=p_0(a)$.
We obtain a profile $S^+$ and $q^+\in F(S^+)$.
For $i\in\{1,\dots,N\}$, let $\Psi_i$ contain the mass clones for $i$ of all parents,
and, for $a'\in A_0\setminus\{c\}$, let $\Phi_{a'}$ contain $a'$ and the padding
clones assigned to $a'$.
These families are pairwise disjoint, every alternative other than $c$ belongs to
exactly one family, and $q^+$ is constant on each family ($\mu_i$ on $\Psi_i$ and
$0$ on each $\Phi_{a'}$).

Successively symmetrize each family; the order is irrelevant.
Let $\tilde R$ denote the resulting profile, so that $q^+\in F(\tilde R)$.
When a family is symmetrized, the majority margin of $c$ against each of its members
becomes the average of the margins of $c$ against the parents of its members in
$R_0$, and it is not affected by the symmetrization of any other family.
Hence, for each $i$ and $u\in\Psi_i$,
\[
  M_{\tilde R}(c,u)
  =\frac{1}{\lVert\nu_i\rVert}\sum_{a\in D}\nu_i(a)M_{R_0}(c,a)>0\text,
\]
and, for $u\in\Phi_{a'}$,
\[
  M_{\tilde R}(c,u)
  =\frac{M_{R_0}(c,a')+m\sum_{a\in D}\nu_1(a)M_{R_0}(c,a)}{1+m\lVert\nu_1\rVert}
  >0
\]
by the choice of $m$, since $M_{R_0}(c,a')\geq -1$.
Thus, $c$ strictly beats every other alternative, and the claim follows by setting
$W=\{c\}$, $B=\Psi_1$, and $\tilde p=q^+$, since $\tilde p(c)=p_0(c)<1$ and
$\tilde p$ is constant and positive on $\Psi_1$.
\end{proof}

Among all $A$, $W$, $B$, $\tilde R$, and $\tilde p$ as in
\Cref{claim:rational-violation}, choose one minimizing $|A^+|$, where
$A^+=\{z\in A\setminus(W\cup B)\colon\tilde p(z)>0\}$.

Since every alternative is a weak Condorcet winner for $\uni(\mathcal L(A))$,
\Cref{lem:condorcet-consistency} and convex-valuedness imply that
$F(\uni(\mathcal L(A)))=\Delta(A)$.
Hence, by population-consistency,
$\tilde p\in F(\lambda\tilde R+(1-\lambda)\uni(\mathcal L(A)))$ for each
$\lambda\in[0,1]\cap\mathbb Q$.
Since this mixture scales all majority margins by $\lambda$,
\ref{item:rational-violation1}--\ref{item:rational-violation3} are preserved for
$\lambda>0$, and $A^+$ is unchanged.
We may thus assume that $\tilde R$ is as close to $\uni(\mathcal L(A))$ as we wish,
and so that $F$ is weakly Condorcet-consistent at $\tilde R$ by
\Cref{lem:condorcet-consistency}.
By \ref{item:rational-violation1}, every $w\in W$ is a weak Condorcet winner for
$\tilde R$, so $w\in F(\tilde R)$, and convex-valuedness yields
$\lambda\,\uni(W)+(1-\lambda)\tilde p\in F(\tilde R)$ for each $\lambda\in[0,1]$.
For $\lambda<1$, replacing $\tilde p$ by such a lottery scales the probabilities on
$A\setminus W$ uniformly, so
\ref{item:rational-violation2} and \ref{item:rational-violation3} persist and $A^+$
is unchanged.

Assume for contradiction that $A^+\neq\emptyset$ and let $a\in A^+$.
Let $\delta=\min\{M_{\tilde R}(w,z)\colon w\in W,\,z\in A\setminus W\}>0$.
Replacing $\tilde p$ by a suitable point on the segment toward $\uni(W)$, we may
assume that $\tilde p(w)/\tilde p(a)=k\in\mathbb N$ with $k>2\delta^{-1}$ for all
$w\in W$.

Successively split each $w\in W$ into a block $E_w$ with $w\in E_w$ and $|E_w|=k$,
distributing $\tilde p(w)$ evenly, so that each member of
$W_0=\bigcup_{w\in W}E_w$ receives probability $\tilde p(w)/k=\tilde p(a)$.
Let $A^*=(A\setminus W)\cup W_0$, and let $R^*$ and $p^*\in F(R^*)$ denote the
resulting profile and lottery.
In $R^*$, the members of $W_0$ are mutually tied and beat every
alternative outside of $W_0$ with a margin of at least $\delta$.
Now symmetrize $W'=W_0\cup\{a\}$, which is possible since $p^*$ is constant
on $W'$, and let $\hat R$ and $\hat p\in F(\hat R)$ denote the result.
In $\hat R$, the members of $W'$ are mutually tied and, for each $u\in W'$ and
$z\in A^*\setminus W'$,
\[
  M_{\hat R}(u,z)
  =\frac{\sum_{v\in W_0}M_{R^*}(v,z)+M_{R^*}(a,z)}{k|W|+1}
  \geq\frac{k|W|\delta-1}{k|W|+1}>0
\]
by the choice of $k$.
Moreover, $\hat p$ is constant on $W'$ with
$\hat p(W')=\tilde p(W)+\tilde p(a)<1$, and the probabilities on $B$ are unchanged.
Hence, $A^*$, $W'$, $B$, $\hat R$, and $\hat p$ satisfy
\ref{item:rational-violation1}--\ref{item:rational-violation3}. The corresponding
set of positive-probability alternatives outside $W'\cup B$ is $A^+\setminus\{a\}$,
contradicting minimality.

Thus, $A^+=\emptyset$, i.e., $\supp(\tilde p)\subseteq W\cup B$.
Moving $\tilde p$ along the segment toward $\uni(W)$ once more, we may assume that
$\tilde p(W)\in(0,1)\cap\mathbb Q$.
Then, $\tilde p$ assigns $\tilde p(W)/|W|\in\mathbb Q$ to each member of $W$ and
$(1-\tilde p(W))/|B|\in\mathbb Q$ to each member of $B$, so
$\tilde p\in\Delta_{\mathbb Q}(A)$.
Since $\tilde p(B)>0$ and $\ml(\tilde R)=\Delta(W)$ by
\ref{item:rational-violation1}, we have $\tilde p\notin\ml(\tilde R)$, which
completes the proof.
\end{proof}

\begin{lemma}\label{lem:f-subseteq-ml}
  Let $F$ be a PSCF that satisfies population-consistency and composition-consistency.
  Then, $F\subseteq\ml$.
\end{lemma}
\begin{proof}
	Assume for contradiction that $F\not\subseteq \ml$.
  By \Cref{lem:rational-violation}, there are $A\in\fone(U)$, $R\in\fprof[A]$, and $p\in F(R)\setminus\ml(R)$ with $p\in\Delta_{\mathbb Q}(A)$.
	  However, \Cref{lem:two-alternatives} together with Lemmas~6 to~13 of \citet{Bran13a} rule out the existence of such a rational violation, a contradiction.
	 
  Therefore, $F\subseteq \ml$.
\end{proof}

\subsection{$\ml\subseteq F$}

We recall Lemma~14 of \citet{Bran13a}.
\begin{lemma}\label{lem:mcgarvey}
	Let $M\in\mathbb{Q}^{n\times n}$ be a skew-symmetric matrix. 
  Then, there are $R\in\fprof[{[n]}]$ and $c\in \mathbb{Q}_{>0}$ such that $c M = M_R$.
	Furthermore, if there is $\pi\in\Pi([n])$ such that $M(i,j) = M(\pi(i),\pi(j))$ for all $i,j\in[n]$, then $R = \pi(R)$.
\end{lemma}

\begin{lemma}\label{lem:ml-subseteq-f}
  Let $F$ be a PSCF that satisfies population-consistency and composition-consistency.
  Then, $\ml\subseteq F$.
\end{lemma}
\begin{proof}
	By \Cref{lem:f-subseteq-ml}, $F\subseteq \ml$.
  Let $A\in\fone(U)$, $R\in\fprof[A]$, and $p\in\ml(R)$.
	Because $F(R)$ is convex and $\ml(R)$ is the convex hull of its vertices, we may assume that $p$ is a vertex of $\ml(R)$; in particular, $p \in \Delta_{\mathbb Q}(A)$, which is the only property of vertices used below.
  By neutrality, we may assume without loss of generality that $A = [n]$ and $\supp(p) = [k]$.
  Moreover, we may assume that $k$ is odd: if $k$ is even, replace $k$ by a component $\{k,k'\}$ in which half of the voters prefer $k$ to $k'$, and split the probability on $k$ equally between $k$ and $k'$. Since $\ml$ satisfies composition-consistency and $k$ and $k'$ are tied, the resulting lottery is a rational maximal lottery of the resulting profile, and its support has odd size. If it is chosen for the resulting profile, composition-consistency of $F$ implies that $p\in F(R)$.\footnote{The original proof by \citet{Bran13a} treats the case of even $k$ separately rather than using this simple reduction argument.}
	It remains to show that $p \in F(R)$.
	
	Consider the case that $k = 1$, i.e., there is $x \in A$ such that $p$ is the degenerate lottery with probability $1$ on $x$.
	Since the degenerate lottery $x$ is a maximal lottery if and only if $x$ is a weak Condorcet winner, $x$ is a weak Condorcet winner for $R$.
	Let $R'$ be a fractional preference profile for which $x$ is a Condorcet winner, and for $\ell \in \mathbb N$, let $R_\ell = \nicefrac1\ell\, R' + \nicefrac{(\ell-1)}{\ell}\, R$.
	Then $x$ is a Condorcet winner for $R_\ell$ for each $\ell\in\mathbb N$.
	Therefore, by \Cref{lem:f-subseteq-ml}, $\emptyset\neq F(R_\ell) \subseteq \ml(R_\ell) = \{p\}$ for each $\ell\in \mathbb N$, and continuity of $F$ implies that $p \in F(R)$.

	Now assume $k \ge 3$.
	By \Cref{lem:mcgarvey}, there are $S\in\fprof[A]$ and $c\in\mathbb{Q}_{>0}$ such that
	\[
	M_S = c
		\left(
	      	\begin{array}{cccccc:ccc}
				0 & -\frac{1}{p_1p_2} & 0 & \dots & 0 & \frac{1}{p_kp_1} & 1 & \dots & 1\\
				\frac{1}{p_1p_2} & & & \multicolumn{2}{c}{\multirow{2}{*}{$\ddots$}} & 0 & \multirow{4}{*}{\vdots} & \multirow{4}{*}{$\ddots$} & \multirow{4}{*}{\vdots}\\
				0 & & \multicolumn{2}{c}{\multirow{2}{*}{$\ddots$}} & & \vdots\\
				\vdots & \multicolumn{2}{c}{\multirow{2}{*}{$\ddots$}} & & & 0\\
				0 & & & & & -\frac{1}{p_{k-1}p_{k}}\\
				-\frac{1}{p_{k}p_{1}} & 0 & \dots & 0 & \frac{1}{p_{k-1}p_k} & 0 & 1 & \dots & 1\\
				\hdashline
				-1 & \multicolumn{4}{c}{\dots} & -1 & 0 & \dots & 0\\
				\vdots & \multicolumn{4}{c}{\ddots} & \vdots & \vdots & \ddots & \vdots\\
				-1 & \multicolumn{4}{c}{\dots} & -1 & 0 & \dots & 0\\
	       	\end{array}
		\right)
	\]
	Intuitively, $M_S$ defines a weighted cycle on $[k]$.
	Note that $(p^\mathsf{T} M_S)_i = 0$ for all $i\in\supp(p)$ and $(p^\mathsf{T} M_S)_i > 0$ for all $i\in A\setminus\supp(p)$, i.e., $p$ is a quasi-strict maximin strategy in $M_S$ in the sense of \citet{Hars73b}.
	Since $p$ is a maximin strategy in $M_S$, it follows that $p\in\ml(S)$.
	For $\epsilon\in[0,1]\cap\mathbb Q$, we define $R^\epsilon = (1-\epsilon) R + \epsilon S$ and $M^\epsilon = M_{R^\epsilon}$.
	Population-consistency of $\ml$ implies that $p\in\ml(R^\epsilon)$ for all $\epsilon\in[0,1]\cap\mathbb Q$.
	Observe that $p$ is a quasi-strict maximin strategy in $M^{\epsilon}$ for every $\epsilon\in(0,1]\cap\mathbb Q$.
	Hence, for every maximin strategy $q$ in $M^{\epsilon}$, it follows that $(q^\mathsf{T} M^{\epsilon})_i = 0$ for every $i\in[k]$ and $q_i = 0$ for every $i\not\in [k]$.
	A direct computation shows that
	\[\det\left(\left(M_S(i,j)\right)_{i,j\in[k-1]}\right) = c^{k-1}\prod_{i = 1}^{k-1} \frac{1}{p_i^2}\neq 0\text,\]
	and hence, $(M_S(i,j))_{i,j\in[k]}$ has rank at least $k-1$. 
	In fact, $(M_S(i,j))_{i,j\in[k]}$ has rank $k-1$, since skew-symmetric matrices of odd size cannot have full rank.\footnote{\label{note:1}A skew-symmetric matrix $M$ of odd size cannot have full rank, since $\det(M) = \det(M^\mathsf{T}) = \det(-M) = (-1)^n\det(M) = -\det(M)$ and, hence, $\det(M) = 0$.}
	Furthermore, $\det((M^{\epsilon}(i,j))_{i,j\in[k-1]})$ is a nonzero polynomial in $\epsilon$ of degree at most $k-1$ and hence, has at most $k-1$ zeros.
	Thus, $(M^{\nicefrac1\ell}(i,j))_{i,j\in[k]}$ has rank $k-1$ for all but finitely many $\ell\in\mathbb{N}$.
	In particular, for all but finitely many $\ell\in\mathbb N$, $(q^\mathsf{T}M^{\nicefrac1\ell})_i = 0$ for all $i\in[k]$ and $q_i = 0$ for all $i\not\in[k]$ imply that $q = p$.
	This implies that $p$ is the unique maximin strategy in $M^{\nicefrac1\ell}$ for all but finitely many $\ell\in\mathbb{N}$ and hence, $\emptyset\neq F(R^{\nicefrac1\ell}) \subseteq \ml(R^{\nicefrac1\ell}) = \{p\}$ for all but finitely many $\ell\in\mathbb{N}$.
	Choose $\ell_0\in\mathbb N$ such that $F(R^{\nicefrac1\ell}) = \{p\}$ for all $\ell\ge\ell_0$, and let $R' = R^{\nicefrac1{\ell_0}}$.
	For each $h\in\mathbb N$, $\nicefrac1h\, R' + (1-\nicefrac1h)\, R = R^{\nicefrac1{h\ell_0}}$, so $F(\nicefrac1h\, R' + (1-\nicefrac1h)\, R) = \{p\}$, and continuity implies that $p\in F(R)$.
\end{proof}

\begin{proof}[Proof of \Cref{thm:main}]
	As discussed in \Cref{sec:results}, \ml is a pSCF that satisfies population-consistency and composition-consistency.
	Conversely, let $f$ be a pSCF that satisfies population-consistency and composition-consistency.
	Let $F$ be the PSCF induced by $f$.
	By \Cref{lem:fractional-extension}, $F$ satisfies population-consistency and composition-consistency.
	By \Cref{lem:f-subseteq-ml} and \Cref{lem:ml-subseteq-f}, $F =\ml$.
	Let $A\in\fone(U)$ and $R\in\prof[A]$ with $\sum_{{\pref}\in\mathcal L(A)} R({\pref}) = n$, and let $R' = \nicefrac1n\, R \in\fprof[A]$.
	Then,
	\[f(R) = F(R') = \ml(R') = \ml(R).\]
	Therefore $f = \ml$.
\end{proof}

\section{Proof of \Cref{thm:alt}}\label{app:proof-alt}

\begin{proof}[Proof of \Cref{thm:alt}]
  $\ml$ satisfies weak Condorcet-consistency: if $x$ is a weak Condorcet winner for $R$, then the degenerate lottery $x$ satisfies $x^\mathsf{T}M_R\geq \mathbf 0$, and therefore $x\in\ml(R)$.
  Moreover, $\ml$ satisfies population-consistency and composition-consistency (see \Cref{sec:results}), and hence cloning-consistency.
  Thus $\ml$ is a rational-valued pSCF satisfying the three axioms in the statement.

  Conversely, let $f$ be a rational-valued pSCF that satisfies population-consistency,
  cloning-consistency, and weak Condorcet-consistency, and let $F$ be the PSCF
  induced by $f$. Then $F$ is rational-valued, and by \Cref{lem:fractional-extension}
  it satisfies population-consistency, cloning-consistency, and weak
  Condorcet-consistency. Moreover, cloning-consistency implies neutrality.
  The argument does not use faithfulness: \Cref{lem:induced-PSCF,lem:fractional-extension}
  apply verbatim to functions that satisfy all conditions of pSCFs except faithfulness,
  in which case $F$ satisfies all conditions of PSCFs except faithfulness, and
  faithfulness of $F$ is not used below. In particular, Lemmas~6--13 of
  \citet{Bran13a} only use unanimity via the two-alternative case, which we
  establish directly.

  We first record the two-alternative case under these assumptions. 
  Let $A=\{x,y\}$ and let $R\in\fprof[A]$. 
  If $R(x,y)=R(y,x)$, then both $x$ and $y$ are weak Condorcet winners. 
  Hence $x,y\in F(R)$, and convex-valuedness gives $F(R)=\Delta(A)=\ml(R)$. 
  Now suppose that $R(x,y)>R(y,x)$, and let $T=\uni(\mathcal L(A))$. 
  Then $F(T)=\Delta(A)$ by the previous case. 
  If $p\in F(R)$ with $p\neq x$, then $p\in F(T)$ as well, and
  population-consistency implies
  \[
    p\in F(\lambda R+(1-\lambda)T)
    \qquad\text{for all }\lambda\in(0,1]\cap\mathbb Q .
  \]
  At every such profile, $x$ is a weak Condorcet winner, so $x\in F(\lambda R+(1-\lambda)T)$. 
  Thus $F$ is not resolute on a relatively open
  interval of two-alternative profiles, contradicting decisiveness. 
  Hence $F(R)=\{x\}=\ml(R)$. 
  The case $R(y,x)>R(x,y)$ is symmetric. 
  Therefore $F=\ml$ on all two-alternative agendas.
  
  Next, we prove that $F\subseteq\ml$. Because $F$ is rational-valued,
  $F(R)=\conv(F(R)\cap\mathbb Q^A)$; as $\ml(R)$ is convex, $F(R)\not\subseteq\ml(R)$
  would entail a \emph{rational} lottery $p\in F(R)\cap\mathbb Q^A$ with
  $p\notin\ml(R)$. It therefore suffices to rule out rational violations, so the
  reduction of \Cref{lem:rational-violation} --- the only step in the proof of
  \Cref{lem:f-subseteq-ml} that relies on composition-consistency --- is not needed
  here. The two-alternative case established above, together with Lemmas~6--13 of
  \citet{Bran13a}, rules out rational violations: as \citet[Remark~5]{Bran13a}
  observe, these lemmas can be carried out with cloning-consistency and weak
  Condorcet-consistency in place of composition-consistency, the appeals to
  \Cref{lem:condorcet-consistency} being superfluous once weak Condorcet-consistency
  is assumed globally; a detailed proof is given by \citet{Bran18a}. Hence
  $F\subseteq\ml$.
  
  It remains to prove the reverse inclusion. Let $A\in\fone(U)$,
  $R\in\fprof[A]$, and $p\in\ml(R)$. Since the vertices of the polytope
  $\ml(R)$ are rational and $F(R)$ is convex, it suffices to show that every
  rational vertex of $\ml(R)$ lies in $F(R)$; fix such a vertex $p$. By
  neutrality, assume without loss of generality that $A=[n]$ and $\supp(p)=[k]$.
  The argument below only uses that $p$ is a rational maximal lottery with
  $\supp(p)=[k]$, not that it is a vertex; in particular, the cloned lottery
  introduced below for even~$k$ need not be a vertex of the cloned profile's
  polytope.
  
  If $k$ is even, add a fresh clone $k'$ of $k$ such that $\{k,k'\}$ is a component
  in which half of the voters prefer $k$ to $k'$, and split the probability of $k$
  equally between $k$ and $k'$. The resulting lottery is still rational and
  has odd support. Moreover, it is a maximal lottery of the cloned profile because
  $\ml$ satisfies composition-consistency and $k$ and $k'$ are tied. Once
  this cloned lottery is shown to be chosen in the cloned profile,
  cloning-consistency projects it back to a lottery in $F(R)$ with the same probabilities
  on all alternatives in $A\setminus\{k\}$; normalization then gives probability
  $p_k$ on $k$. Hence it is enough to consider the case where $k$ is odd.
  If $k=1$, then $p=x$ for some $x\in A$. Since $x\in\ml(R)$, $x$ is a weak
  Condorcet winner for $R$, and weak Condorcet-consistency gives $p=x\in F(R)$.
  Finally, suppose that $k\geq 3$ is odd. The perturbation argument in the proof
  of \Cref{lem:ml-subseteq-f}, based on \Cref{lem:mcgarvey}, uses only
  population-consistency, continuity, neutrality, and the already established
  inclusion $F\subseteq\ml$. It does not use composition-consistency. Hence the
  same argument applies here.
  
  We have shown that $F=\ml$. 
  As in the proof of \Cref{thm:main}, it follows that $f=\ml$.
  \end{proof}


\begin{thebibliography}{55}
\providecommand{\natexlab}[1]{#1}
\providecommand{\url}[1]{\texttt{#1}}
\expandafter\ifx\csname urlstyle\endcsname\relax
  \providecommand{\doi}[1]{doi: #1}\else
  \providecommand{\doi}{doi: \begingroup \urlstyle{rm}\Url}\fi

\bibitem[Berker et~al.(2025)Berker, Casacuberta, Robinson, Ong, Conitzer, and
  Elkind]{BCR+25a}
R.~E. Berker, S.~Casacuberta, I.~Robinson, C.~Ong, V.~Conitzer, and E.~Elkind.
\newblock From independence of clones to composition consistency: A hierarchy
  of barriers to strategic nomination.
\newblock In \emph{Proceedings of the 26th ACM Conference on Economics and
  Computation (ACM-EC)}, 2025.

\bibitem[Brandl(2018)]{Bran18a}
F.~Brandl.
\newblock \emph{Zero-Sum Games in Social Choice and Game Theory}.
\newblock PhD thesis, Technische Universit{\"a}t M{\"u}nchen, 2018.

\bibitem[Brandl and Brandt(2019)]{BrBr17c}
F.~Brandl and F.~Brandt.
\newblock Justifying optimal play via consistency.
\newblock \emph{Theoretical Economics}, 14\penalty0 (4):\penalty0 1185--1201,
  2019.

\bibitem[Brandl and Brandt(2020)]{BrBr17a}
F.~Brandl and F.~Brandt.
\newblock Arrovian aggregation of convex preferences.
\newblock \emph{Econometrica}, 88\penalty0 (2):\penalty0 799--844, 2020.

\bibitem[Brandl and Brandt(2024{\natexlab{a}})]{BrBr21a}
F.~Brandl and F.~Brandt.
\newblock A natural adaptive process for collective decision-making.
\newblock \emph{Theoretical Economics}, 19\penalty0 (2):\penalty0 667--703,
  2024{\natexlab{a}}.

\bibitem[Brandl and Brandt(2024{\natexlab{b}})]{BrBr23a}
F.~Brandl and F.~Brandt.
\newblock An axiomatic characterization of {N}ash equilibrium.
\newblock \emph{Theoretical Economics}, 19\penalty0 (4):\penalty0 1473--1504,
  2024{\natexlab{b}}.

\bibitem[Brandl and Peters(2019)]{BrPe17a}
F.~Brandl and D.~Peters.
\newblock An axiomatic characterization of the {B}orda mean rule.
\newblock \emph{Social Choice and Welfare}, 52\penalty0 (4):\penalty0 685--707,
  2019.

\bibitem[Brandl and Peters(2022)]{BrPe19a}
F.~Brandl and D.~Peters.
\newblock Approval voting under dichotomous preferences: {A} catalogue of
  characterizations.
\newblock \emph{Journal of Economic Theory}, 205, 2022.

\bibitem[Brandl et~al.(2016)Brandl, Brandt, and Seedig]{Bran13a}
F.~Brandl, F.~Brandt, and H.~G. Seedig.
\newblock Consistent probabilistic social choice.
\newblock \emph{Econometrica}, 84\penalty0 (5):\penalty0 1839--1880, 2016.

\bibitem[Brandl et~al.(2022)Brandl, Brandt, and Stricker]{BBS20a}
F.~Brandl, F.~Brandt, and C.~Stricker.
\newblock An analytical and experimental comparison of maximal lottery schemes.
\newblock \emph{Social Choice and Welfare}, 58\penalty0 (1):\penalty0 5--38,
  2022.

\bibitem[Brandl et~al.(2026)Brandl, Brandt, and Seedig]{BBS15a}
F.~Brandl, F.~Brandt, and H.~G. Seedig.
\newblock Consistent probabilistic social choice.
\newblock Technical report, https://arxiv.org/abs/1503.00694, 2026.

\bibitem[Brandt(2011)]{Bran11b}
F.~Brandt.
\newblock Minimal stable sets in tournaments.
\newblock \emph{Journal of Economic Theory}, 146\penalty0 (4):\penalty0
  1481--1499, 2011.

\bibitem[Brandt(2017)]{Bran17a}
F.~Brandt.
\newblock Rolling the dice: {R}ecent results in probabilistic social choice.
\newblock In U.~Endriss, editor, \emph{Trends in Computational Social Choice},
  chapter~1, pages 3--26. AI Access, 2017.

\bibitem[Brandt(2026)]{Bran22a}
F.~Brandt.
\newblock Stochastic choice and dynamics based on pairwise comparisons.
\newblock In M.~Voorneveld, J.~W. Weibull, T.~Andersson, R.~B. Myerson, J.-F.
  Laslier, R.~Laraki, and Y.~Koriyama, editors, \emph{One Hundred Years of Game
  Theory: {A} {N}obel Symposium}, Econometric Society Monographs, pages
  158--166. Cambridge University Press, 2026.

\bibitem[Brandt et~al.(2011)Brandt, Brill, and Seedig]{BBS11a}
F.~Brandt, M.~Brill, and H.~G. Seedig.
\newblock On the fixed-parameter tractability of composition-consistent
  tournament solutions.
\newblock In \emph{Proceedings of the 22nd International Joint Conference on
  Artificial Intelligence (IJCAI)}, pages 85--90, 2011.

\bibitem[Brandt et~al.(2025)Brandt, Dong, and Peters]{BDP24a}
F.~Brandt, C.~Dong, and D.~Peters.
\newblock Condorcet-consistent choice among three candidates.
\newblock \emph{Games and Economic Behavior}, 153:\penalty0 113--130, 2025.

\bibitem[Congar and Merlin(2012)]{CoMe12a}
R.~Congar and V.~R. Merlin.
\newblock A characterization of the maximin rule in the context of voting.
\newblock \emph{Theory and Decision}, 72\penalty0 (1):\penalty0 131--147, 2012.

\bibitem[Dasgupta and Maskin(2008)]{DaMa08a}
P.~Dasgupta and E.~Maskin.
\newblock On the robustness of majority rule.
\newblock \emph{Journal of the European Economic Association}, 6\penalty0
  (5):\penalty0 949--973, 2008.

\bibitem[Dutta and Laslier(1999)]{DuLa99a}
B.~Dutta and J.-F. Laslier.
\newblock Comparison functions and choice correspondences.
\newblock \emph{Social Choice and Welfare}, 16\penalty0 (4):\penalty0 513--532,
  1999.

\bibitem[Felsenthal and Machover(1992)]{FeMa92a}
D.~S. Felsenthal and M.~Machover.
\newblock After two centuries should {C}ondorcet's voting procedure be
  implemented?
\newblock \emph{Behavioral Science}, 37\penalty0 (4):\penalty0 250--274, 1992.

\bibitem[Fine and Fine(1974)]{FiFi74a}
B.~Fine and K.~Fine.
\newblock Social choice and individual ranking {I}.
\newblock \emph{Review of Economic Studies}, 41\penalty0 (127):\penalty0
  303--323, 1974.

\bibitem[Fishburn(1978)]{Fish78d}
P.~C. Fishburn.
\newblock Axioms for approval voting: {D}irect proof.
\newblock \emph{Journal of Economic Theory}, 19\penalty0 (1):\penalty0
  180--185, 1978.

\bibitem[Fishburn(1984)]{Fish84a}
P.~C. Fishburn.
\newblock Probabilistic social choice based on simple voting comparisons.
\newblock \emph{Review of Economic Studies}, 51\penalty0 (4):\penalty0
  683--692, 1984.

\bibitem[Fisher and Ryan(1995)]{FiRy95a}
D.~C. Fisher and J.~Ryan.
\newblock Tournament games and positive tournaments.
\newblock \emph{Journal of Graph Theory}, 19\penalty0 (2):\penalty0 217--236,
  1995.

\bibitem[Harsanyi(1973)]{Hars73b}
J.~C. Harsanyi.
\newblock Games with randomly disturbed payoffs: A new rationale for
  mixed-strategy equilibrium points.
\newblock \emph{International Journal of Game Theory}, 2\penalty0 (1):\penalty0
  1--23, 1973.

\bibitem[Holliday(2024)]{Holl24a}
W.~H. Holliday.
\newblock An impossibility theorem concerning positive involvement in voting.
\newblock \emph{Economics Letters}, 236:\penalty0 111589, 2024.

\bibitem[Holliday and Pacuit(2023)]{HoPa20a}
W.~H. Holliday and E.~Pacuit.
\newblock Split cycle: {A} new {C}ondorcet-consistent voting method independent
  of clones and immune to spoilers.
\newblock \emph{Public Choice}, 197\penalty0 (1--2):\penalty0 1--62, 2023.

\bibitem[Horan(2013)]{Hora13a}
S.~M. Horan.
\newblock Implementation of majority voting rules.
\newblock 2013.
\newblock Working paper.

\bibitem[Kreweras(1965)]{Krew65a}
G.~Kreweras.
\newblock Aggregation of preference orderings.
\newblock In \emph{Mathematics and Social Sciences {I}: Proceedings of the
  seminars of {Menthon-Saint-Bernard}, {F}rance (1--27 July 1960) and of
  G{\"o}sing, {A}ustria (3--27 July 1962)}, pages 73--79, 1965.

\bibitem[Lackner and Skowron(2021)]{LaSk21a}
M.~Lackner and P.~Skowron.
\newblock Consistent approval-based multi-winner rules.
\newblock \emph{Journal of Economic Theory}, 192:\penalty0 105173, 2021.

\bibitem[Laffond et~al.(1993)Laffond, Laslier, and {Le Breton}]{LLL93b}
G.~Laffond, J.-F. Laslier, and M.~{Le Breton}.
\newblock The bipartisan set of a tournament game.
\newblock \emph{Games and Economic Behavior}, 5\penalty0 (1):\penalty0
  182--201, 1993.

\bibitem[Laffond et~al.(1996)Laffond, Lain{\'e}, and Laslier]{LLL96a}
G.~Laffond, J.~Lain{\'e}, and J.-F. Laslier.
\newblock Composition-consistent tournament solutions and social choice
  functions.
\newblock \emph{Social Choice and Welfare}, 13\penalty0 (1):\penalty0 75--93,
  1996.

\bibitem[Laffond et~al.(1997)Laffond, Laslier, and {Le Breton}]{LLL97a}
G.~Laffond, J.-F. Laslier, and M.~{Le Breton}.
\newblock A theorem on symmetric two-player zero-sum games.
\newblock \emph{Journal of Economic Theory}, 72\penalty0 (2):\penalty0
  426--431, 1997.

\bibitem[Laslier(1996)]{Lasl96a}
J.-F. Laslier.
\newblock Rank-based choice correspondences.
\newblock \emph{Economics Letters}, 52\penalty0 (3):\penalty0 279--286, 1996.

\bibitem[Laslier(1997)]{Lasl97a}
J.-F. Laslier.
\newblock \emph{Tournament Solutions and Majority Voting}.
\newblock Springer-Verlag, 1997.

\bibitem[Laslier(2000)]{Lasl00a}
J.-F. Laslier.
\newblock Aggregation of preferences with a variable set of alternatives.
\newblock \emph{Social Choice and Welfare}, 17\penalty0 (2):\penalty0 269--282,
  2000.

\bibitem[Lederer(2024)]{Lede23a}
P.~Lederer.
\newblock Bivariate scoring rules: Unifying the characterizations of positional
  scoring rules and {K}emeny's rule.
\newblock \emph{Journal of Economic Theory}, 218:\penalty0 105836, 2024.

\bibitem[Myerson(1995)]{Myer95b}
R.~B. Myerson.
\newblock Axiomatic derivation of scoring rules without the ordering
  assumption.
\newblock \emph{Social Choice and Welfare}, 12\penalty0 (1):\penalty0 59--74,
  1995.

\bibitem[Nehring and Pivato(2018)]{NePi18a}
K.~Nehring and M.~Pivato.
\newblock The median rule in judgement aggregation.
\newblock 2018.
\newblock Working paper.

\bibitem[{\"O}zt{\"u}rk(2020)]{OzTu20a}
Z.~E. {\"O}zt{\"u}rk.
\newblock Consistency of scoring rules: {A} reinvestigation of
  composition-consistency.
\newblock \emph{International Journal of Game Theory}, 49:\penalty0 801--831,
  2020.

\bibitem[Pivato(2013)]{Piva13a}
M.~Pivato.
\newblock Variable-population voting rules.
\newblock \emph{Journal of Mathematical Economics}, 49\penalty0 (3):\penalty0
  210--221, 2013.

\bibitem[Rivest and Shen(2010)]{RiSh10a}
R.~L. Rivest and E.~Shen.
\newblock An optimal single-winner preferential voting system based on game
  theory.
\newblock In \emph{Proceedings of the 3rd International Workshop on
  Computational Social Choice (COMSOC)}, pages 399--410, 2010.

\bibitem[Saari(1990)]{Saar90a}
D.~G. Saari.
\newblock Consistency of decision processes.
\newblock \emph{Annals of Operations Research}, 23\penalty0 (1):\penalty0
  103--137, 1990.

\bibitem[Saari(1995)]{Saar95a}
D.~G. Saari.
\newblock \emph{Basic Geometry of Voting}.
\newblock Springer, 1995.

\bibitem[Skowron et~al.(2019)Skowron, Faliszewski, and Slinko]{SFS19a}
P.~Skowron, P.~Faliszewski, and A.~Slinko.
\newblock Axiomatic characterization of committee scoring rules.
\newblock \emph{Journal of Economic Theory}, 180:\penalty0 244--273, 2019.

\bibitem[Smith(1973)]{Smit73a}
J.~H. Smith.
\newblock Aggregation of preferences with variable electorate.
\newblock \emph{Econometrica}, 41\penalty0 (6):\penalty0 1027--1041, 1973.

\bibitem[Tideman(1986)]{Tide86a}
T.~N. Tideman.
\newblock A majority-rule characterization with multiple extensions.
\newblock \emph{Social Choice and Welfare}, 3\penalty0 (1):\penalty0 17--30,
  1986.

\bibitem[Tideman(1987)]{Tide87a}
T.~N. Tideman.
\newblock Independence of clones as a criterion for voting rules.
\newblock \emph{Social Choice and Welfare}, 4\penalty0 (3):\penalty0 185--206,
  1987.

\bibitem[{von Neumann}(1928)]{vNeu28a}
J.~{von Neumann}.
\newblock Zur {T}heorie der {G}esellschaftsspiele.
\newblock \emph{Mathematische Annalen}, 100\penalty0 (1):\penalty0 295--320,
  1928.

\bibitem[Young(1974{\natexlab{a}})]{Youn74a}
H.~P. Young.
\newblock An axiomatization of {B}orda's rule.
\newblock \emph{Journal of Economic Theory}, 9\penalty0 (1):\penalty0 43--52,
  1974{\natexlab{a}}.

\bibitem[Young(1974{\natexlab{b}})]{Youn74b}
H.~P. Young.
\newblock A note on preference aggregation.
\newblock \emph{Econometrica}, 42\penalty0 (6):\penalty0 1129--1131,
  1974{\natexlab{b}}.

\bibitem[Young(1975)]{Youn75a}
H.~P. Young.
\newblock Social choice scoring functions.
\newblock \emph{SIAM Journal on Applied Mathematics}, 28\penalty0 (4):\penalty0
  824--838, 1975.

\bibitem[Young(1977)]{Youn77a}
H.~P. Young.
\newblock Extending {C}ondorcet's rule.
\newblock \emph{Journal of Economic Theory}, 16:\penalty0 335--353, 1977.

\bibitem[Young and Levenglick(1978)]{YoLe78a}
H.~P. Young and A.~B. Levenglick.
\newblock A consistent extension of {C}ondorcet's election principle.
\newblock \emph{SIAM Journal on Applied Mathematics}, 35\penalty0 (2):\penalty0
  285--300, 1978.

\bibitem[Zavist and Tideman(1989)]{ZaTi89a}
T.~M. Zavist and T.~N. Tideman.
\newblock Complete independence of clones in the ranked pairs rule.
\newblock \emph{Social Choice and Welfare}, 6\penalty0 (2):\penalty0 167--173,
  1989.

\end{thebibliography}
\end{document}